\documentclass[a4paper,UKenglish,cleveref,thm-restate]{lipics-v2021}

\title{Verification of Correlated Equilibria in Concurrent Games}
\usepackage{xparse}

\author{Senthil Rajasekaran}
{Université Libre de Bruxelles, Brussels, Belgium}
{senthilrajasekaran@ulb.be}
{https://orcid.org/0000-0002-6675-8063}
{}
\author{Jean-François Raskin}
{Université Libre de Bruxelles, Brussels, Belgium}
{jean-francois.raskin@ulb.be}
{https://orcid.org/0000-0002-3673-1097}
{}
\author{Moshe Y. Vardi}
{Rice University, Houston, USA}
{vardi@rice.edu}
{https://orcid.org/0000-0002-0661-5773}
{}

\authorrunning{S. Rajasekaran, J.F. Raskin, and M.Y. Vardi}
\Copyright{S. Rajasekaran, J.F. Raskin, and M.Y. Vardi}
\newtheorem*{manualtheorem@base}{Theorem}
\newtheorem*{manuallemma@base}{Lemma}

\NewDocumentEnvironment{manualtheorem}{m o}{%
  \par\medskip
  \noindent
  \textbf{Theorem~#1.}%
  \IfValueT{#2}{\ (One-Step Deviation)}%
  \;\itshape
}{%
  \par\medskip
}

\NewDocumentEnvironment{manuallemma}{m o}{%
  \par\medskip
  \noindent\textbf{Lemma~#1%
  \IfValueT{#2}{\ (#2)}.}\;\itshape
}{%
  \par\medskip
}


\usepackage{amsmath, amssymb, amsthm}
\usepackage{tikz}
\usetikzlibrary{automata, positioning, arrows, shapes, matrix, decorations.pathreplacing}
\tikzset{
->,
node distance=3.5cm,
every state/.style={thick, fill=gray!10},
initial text=$ $,
}
\nolinenumbers

\newtheorem*{problem*}{Decision Problem}

\usepackage{etoolbox} 

\hypersetup{
    colorlinks=true,
    linkcolor=blue,
    filecolor=magenta,      
    urlcolor=cyan
}

\newcommand{\jfr}[1]{{{\color{orange}{[JFR:]} \color{blue}#1}}}
\newcommand{\sr}[1]{{{\color{purple}{[SR:]} \color{blue}#1}}}

\title{Verification of Correlated Equilibria in Concurrent Reachability Games}

\begin{document}

\maketitle
\begin{abstract}As part of an effort to apply the rigorous guarantees of formal verification to multi-agent systems, the field of equilibrium analysis, also called rational verification, studies equilibria in multiplayer games to reason about system-level properties such as safety and scalability. While most prior work focuses on deterministic settings, recent probabilistic extensions enable the use of richer equilibrium concepts. In this paper, we study one such equilibrium concept -- correlated equilibria -- and introduce a natural refinement-- subgame-perfect correlated equilibria -- in the context of the verification problem. We characterize the computational complexity of verifying such equilibria and show a somewhat surprising separation (under standard complexity-theoretic assumptions): despite being more general, correlated equilibria yield a strictly harder P-complete verification problem than the subgame-perfect correlated equilibria verification problem, which can be solved in log-squared-space. We further analyze the setting where inputs are given succinctly via Bayesian networks, as the study of succinct representations is an important direction to connect static complexity-theoretic analysis to real-world program representations, and show that this complexity gap disappears under such representations.\end{abstract}

\section{Introduction}\label{introduction}

Computer systems in which multiple autonomous agents concurrently interact with one another are nearly ubiquitous in the design and construction of modern computer systems~\cite{SLmultiagentbook,gullí2025agentic}. When multiple autonomous agents interact concurrently, the total number of possible interactions increases exponentially with the number of agents~\cite{SLmultiagentbook}. While these interactions give a system a large amount of expressive power, they also introduce the possibility of undesirable behaviors, such as violations of the system's overall safety or the creation of highly inefficient bottlenecks~\cite{principlesofmodelcheckingbook,AGHKNPSW21,rajasekaranthesis}.

This has motivated researchers in the field of formal verification to study concurrent multi-agent systems, with the broad goal of providing rigorous theoretical guarantees about the overall behavior of concurrent multi-agent systems in which each autonomous agent acts rationally. To provide a rigorous mathematical framework for studying these problems, researchers have studied the concept of ``equilibria'', which are solution concepts for multi-player games from the field of game theory~\cite{Osborne1994}. This has led to the creation of the field of ``equilibrium analysis'', which is also called rational verification~\cite{AGHKNPSW21}, in which researchers study the complexity of decision problems related to equilibria in multi-agent systems, an idea that has led to many recent research works, e.g.~\cite{GPW17,parkerverification,raskinsubgame2,PRISM,EVE,rajasekaranthesis}.

The field of equilibrium analysis has a strong preference for systems in which both the transitions and the agents are deterministic. This has led to a preference for solution concepts that work well in the deterministic setting, such as the Nash equilibrium, see~\cite{RV21,RV22,RBV23,RV26,GutierrezNPW20,iBG} for a few examples. To reason about more general systems, the literature has begun to consider different equilibrium types and settings. For example, new equilibrium concepts such as subgame-perfect equilibria have been considered in the deterministic setting~\cite{raskinsubgame1,raskinsubgame2}. More recently, the deterministic setting has been expanded to general probabilistic settings, e.g.~\cite{RV25,parkerverification}. The introduction of probabilistic settings not only represents a crucial step towards the general goal of modeling general concurrent systems, but it also allows for the use of equilibrium concepts that only ``truly'' exist in probabilistic settings.

The concept of \emph{correlated equilibrium}, which is one of the main objects of study in this paper, is one such example of an equilibrium concept that is only sensible in the probabilistic setting. To see this, let us first consider the concept of Nash equilibria in probabilistic multi-agent systems. A Nash equilibrium is a tuple of probabilistic \emph{strategies} -- each agent is assigned their own individual strategy that determines their behavior in the system.  A Nash equilibrium can be viewed as a central strategy,  which we call a \emph{controller}, that consists of the cross-product of all of the agents' individual strategies. In this way, the Nash equilibrium induces a product distribution on the set of joint agent actions~\cite{Osborne1994,SLmultiagentbook}. Therefore, in a Nash equilibrium, agent actions are \emph{uncorrelated}.

In contrast, in a correlated equilibrium, the controller samples from a distribution over tuples of joint action choices. Thus, the distribution is not necessarily built from the product of other distributions, allowing for correlation between agent actions. The controller then privately informs each agent of which action they alone have been recommended (and these recommendations are referred to as \emph{controller advice}). So, for example, if the controller samples the action tuple ``abc'', it would only inform the second agent to play `b', and not inform the second agent that the first agent is informed to play `a' or that the third agent is informed to play `c'. By allowing a central planner to sample from a distribution over the space of joint actions, the controller can correlate agent action choices. 

Such correlation of actions is an especially useful concept in settings where agents behave independently but share a common system-level objective. Consequently, correlated equilibria have seen applications in networking (see~\cite{correlatednetworkexample} for a survey particularly focusing on \emph{cognitive radio networks}), smart-grid management (see~\cite{correlatedgridexample} for a recent example), and other domains. 
A correlated equilibrium corresponds to a scenario in which a central controller recommends actions to agents individually. This makes it well-suited to the field of formal methods, where the setting often involves a system with multiple independent agents that are not necessarily antagonistic to a centralized system-level goal~\cite{principlesofmodelcheckingbook,AGHKNPSW21}, and are therefore amenable to receiving instructions from a centralized controller. Furthermore, there are also mathematical benefits to considering correlation.   It is well known that the set of correlated equilibria generally contains the set of Nash equilibria as a subset~\cite{Osborne1994}. Consequently, there may be system-level desired outcomes that can be sustained as correlated equilibria but cannot be realized by a Nash equilibrium. For interested readers, we review such an example in Section~\ref{examplesection}.

Despite being well-suited to formal methods and, therefore, to equilibrium analysis, the equilibrium-analysis literature's preference for deterministic settings means that correlated equilibria have not been studied in depth. As previously mentioned, this concept is viable only in probabilistic settings. This is because the interesting dynamics of correlated equilibria arise when agents can reason about the joint action via Bayesian inference. For example, by recommending the second agent `b', as in the example above, the agent may infer that certain joint actions, such as ``aaa'', could not have been sampled by the controller and update their beliefs accordingly. This is why correlated equilibria require a probabilistic setting, as a controller that behaves deterministically and therefore gives deterministic recommendations gives agents complete information about the joint action chosen, removing the need for Bayesian inference.

The two central decision problems in equilibrium analysis are \emph{verification} and \emph{realizability}: verification asks whether a given specification of player behavior is an equilibrium in a given system, while realizability asks whether an equilibrium exists at all. In the probabilistic setting, both coincide with deep questions of game theory, and the realizability problem in particular has resisted progress for decades~\cite{challenge}. The underlying difficulty is that undiscounted, infinite-horizon games in which every player has a reachability goal do not admit a continuous payoff function (see Section~\ref{discont} in the Appendix), so the fixed-point arguments of~\cite{Nash51} cannot be applied. Because the mathematical structure we deal with is much closer to the game-theory literature -- where probabilistic models are more standard than in equilibrium analysis -- we adopt game-theoretic terminology throughout, referring to ``systems'' as ``games'' and ``agents'' as ``players''.

The complexity of the realizability problem for correlated equilibria -- which, given a game $G$, asks whether there is a controller advice $D$ that constitutes a correlated equilibrium in $G$ -- is itself open under the definitions we adopt. Every Nash equilibrium is, in particular, a correlated equilibrium~\cite{Osborne1994}, and whether Nash equilibria always exist in games in which each player has a reachability goal is a long-standing open problem, a claimed resolution in~\cite{chatterjee-majumdar-jurdzinski-csl2004} notwithstanding. The controller advice in our formulation is \emph{memoryless} -- the recommended distribution at a state depends on that state alone, not on the history of play -- which is, we believe, the most natural way to apply the definition of a correlated equilibrium to the probabilistic concurrent setting. This restriction is not innocuous: there is a three-player turn-based game that admits no memoryless Nash equilibrium (Proposition~3.13 of~\cite{ummelsthesis}; see also~\cite{nostationaryNEkuipers}). Consequently, the realizability of correlated equilibria is non-trivial in our setting, and many deep questions remain open.

In this paper, we study the complexity of the verification problem for correlated equilibria and for a natural strengthening, \emph{subgame-perfect} correlated equilibria. We pay close attention to the low-level details of numerical representation, particularly for probabilities, in order to obtain sharp complexity-theoretic statements. Our main results exhibit a somewhat counterintuitive separation that mirrors the one reported in~\cite{RV25}: under standard complexity-theoretic assumptions, verification turns out to be \emph{easier} for the subgame-perfect variant, as we show that it can be done in LOG\textsuperscript{2}SPACE (Theorem~\ref{subgameverificationcomplexitytheorem}) than for the plain correlated equilibrium, for which we establish P-Completeness (Theorem~\ref{corrverificationcompletenesstheorem}). Because representation plays such a central role in these results, we also consider succinct presentations of multi-player games via Bayesian networks and analyze how this change of input format affects the complexity of verification for both equilibrium concepts. In the succinct representation case, we show that the complexity of the two verification problems coincides as they are shown to be PP-hard and solvable with polynomially many calls to a PP oracle (Theorem~\ref{BNrepcomplexitytheorem}).

\section{Preliminaries}\label{definitionsection}

Our main objects of study are the \emph{probabilistic concurrent game} and the \emph{controller advice}.

\begin{definition}\label{pCGS}

A {\em probabilistic concurrent game graph} (henceforth, just \emph{game}) is  a 7-tuple $G=\langle  Q, q_{\sf init},\Pi,A,{\cal A},\delta,(R_i)_{i\in \Pi} \rangle$ specified by the following components:
    {\bf (1)} A finite non-empty set of states $Q$, with $q_{\sf init} \in Q$ the initial state.
    {\bf (2)} A finite non-empty set of players $\Pi$.
    {\bf (3)} A finite non-empty set of action names $A$. Given the action set $A$ and the set of players $\Pi$, we define the set of \emph{joint actions} $A^{\Pi}$. Players in a game choose actions concurrently, and a joint action is the result of $\Pi$ players all selecting one of the $A$ actions concurrently.  For the sake of convenience, we also introduce the notation $A_i \subseteq A$ as the set of actions associated with Player~$i$ specifically, which may be restricted by $\mathcal{A}$, as explained in the next point.
    {\bf (4)} $\mathcal{A}\ : Q \times \Pi \rightarrow 2^A \setminus \emptyset$, a function that maps each state $q$ and player $i \in \Pi$ to the non-empty subset of actions ${\cal A}(q,i) \subseteq A$ available to Player~$i$ in state $q$.
    We assume that joint actions always respect ${\cal A}$. 
    {\bf (5)} $\delta:Q \times A^{\Pi}  \rightarrow {\sf Dist}(Q)$ is the probabilistic transition relation, a partial function that maps states and joint actions to a distribution over states representing potential successor states.
    {\bf (6)} The set $R_i \subseteq Q $ is the \emph{reachability goal} of Player~$i$. 
\end{definition}
The game starts in state $q_{\sf init} \in Q$. At each state, every player chooses an action from the set of action names $A$. Given this joint choice, an element of $A^{\Pi}$, the transition function $\delta$ determines a distribution over the possible next states. This distribution is sampled, and the game moves to the resulting state, where the process repeats. This results in an infinite sequence of states (which we call a \emph{trace}) $t \in Q^{\omega}$. Given a trace $t$, Player~$i$ receives a payoff of $1$ if the infinite trace contains a state $q \in R_i$ and $0$ otherwise.

The transition function $\delta$ is probabilistic and provided as part of the input. We therefore say that the probabilities in $\delta$ are at most $\ell$ bits in length, where $\ell$ is the size of the largest number represented in $\delta$. We go into more detail about the representation of $\delta$ after introducing the \emph{controller advice} $D$. As described in Section~\ref{introduction}, the controller advice is the method by which a central planner recommends players actions to play. 

\begin{definition}\label{controlleradvicedefinition}
    A \emph{controller advice} $D : Q \rightarrow {\sf Dist}(A^{\Pi})$ maps states to a joint distribution over joint actions. We assume that all controller advices considered are \emph{``valid''}, i.e. they never suggest a tuple with positive probability that instructs a Player~$i$ to play an action restricted by $\mathcal{A}$.
\end{definition}

A controller advice $D$ works by sampling from its joint distribution and showing each player their sampled action. Players may view their own suggested action and the total joint distribution from which they are sampled, but not the actions suggested to the other players. This means that players may be able to update their knowledge of the sampling distribution based on the action they are given. For example, suppose we are considering a three-player game $G$ in which each player has two actions: $a$ and $b$. At a certain state $q$, we have $D(q) = \{ aab : .5 , bbb : .5 \}$. If player $1$ or $2$ is recommended to play $b$ at this state, he can conclude with certainty that the joint action sampled was $bbb$, as this is the only possible tuple that is consistent with the recommended action $b$. Player 3, however, is always recommended action $b$, and so she gets no extra information upon viewing a recommendation of $b$. Note that player 3 can similarly deduce the joint distribution of players 1 and 2, which is a 50/50 distribution over $aa$ and $bb$.

Just as with $\delta$, we assume that the probabilities in the controller are represented by at most $\ell$ bits. This general parameter $\ell$ can therefore be thought of as the maximum of the largest representation of a probability in $\delta$ and the largest representation of a probability in $D$. Critically, the bound $\ell$ is derived from the fact that we are given $\delta$ and $D$ as input; it is not an a priori restriction. 

We represent the controller advice $D$ as a set of $|Q|$ binary tables of the following form, with tuples in $A^{\Pi}$ listed on the left column and (at most) $\ell$-length numerals listed on the right column.
\begin{figure}[ht]
\centering
\begin{tikzpicture}[
  every node/.style={font=\normalsize},
  scale=0.8,
  transform shape,
  brace style/.style={decorate,decoration={brace,amplitude=6pt,raise=2pt}}
]
\newlength{\RightColContentWidth}
\newlength{\RightColWidth}
\settowidth{\RightColContentWidth}{$0.58214\cdots$}
\setlength{\RightColWidth}{\dimexpr\RightColContentWidth + 2\dimexpr\tabcolsep\relax + 2\dimexpr\arrayrulewidth\relax\relax}

\node[inner sep=0pt,outer sep=0pt,rectangle] (tbl) {
  \begin{tabular}{|c|c|}
  \hline
  $a_1 a_2 a_3$ & $0.25192\cdots$ \\ \hline
  $a_1 a_2 a_4$ & $0.11723\cdots$ \\ \hline
  \vdots        & \vdots          \\ \hline
  $a_{9} a_{4} a_{2}$ & $0.16723\cdots$ \\ \hline
  \vdots        & \vdots          \\ \hline
  \end{tabular}
};

\coordinate (RC-NE) at ([yshift=4pt]tbl.north east);
\coordinate (RC-NW) at ([xshift=-\RightColWidth]RC-NE);
\draw[brace style, -] (RC-NW) -- (RC-NE) node[midway,above=6pt] {$ \leq \ell$ bits};

\draw[decorate,decoration={brace,amplitude=6pt,raise=2pt,mirror}, -]
    ([xshift=-8pt]tbl.north west) --
    ([xshift=-8pt]tbl.south west)
    node[midway,left=8pt] {$\le T$ rows};

\end{tikzpicture}

\caption{A depiction of the representation of $D$ for a single state $q \in Q$. This table specifies the probability distribution $D(q)$ over $A^{\Pi}$ through a table that matches tuples in $A^{\Pi}$ to numerals that have representations that are at most $\ell$ bits long. Action tuples that are played with zero probability are simply omitted from the table, meaning that it is possible for the table to have fewer than $|A^{\Pi}|$ rows. Specifically, we say that the largest such table over all $q \in Q$ has $T$ rows, meaning that each table has $\leq T$ rows definitionally.}
\label{controllerrepresentation}
\end{figure}
The transition function $\delta: Q \times A^{\Pi} \rightarrow {\sf Dist}(Q)$ is represented in a similar way, but instead of matching tuples to $\ell$-length probabilities, it matches triples $\langle q, a, q' \rangle$ with $q,q' \in Q$ and $a \in A^{\Pi}$ to $\ell$-length probabilities. Thus, there are two additional columns that specify the starting and ending states. As before, tuples that are associated with $0$ probability are omitted from the representation of the transition table, meaning that it is once again possible for $\delta$ to have fewer than $|A^{\Pi}| \cdot |Q|^2$ rows. An alternative, more succinct representation of $\delta$ and $D$ based on \emph{Bayesian Networks} is considered in the latter half of the paper.

As mentioned earlier, the output of $D$ should be intuitively viewed as a ``recommendation''. If a player does not wish to follow the controller advice, they may instead choose their actions in a different manner by employing a \emph{strategy}.

\begin{definition}\label{strategydef}
    A \emph{strategy} for Player~$i$ is a function $\pi_i$ that takes as input a state $q$ and an action $ a \in \mathcal{A}(q,i)$ and returns a distribution over the actions allowed by $\mathcal{A}$.
    Formally, $\pi_i : (Q \times A_i)^* \rightarrow {\sf Dist} (A_i)$. Just as $D$ respects $\mathcal{A}$, so must $\pi_i$ by never recommending forbidden actions.

\end{definition}
Note that there is no a priori bound on the length of the numerals used in a strategy. The set of all Player~$i$ strategies is denoted $\Pi_i$. Note that player strategies take a sequence of actions as input. This represents a player potentially taking into account the history of the controller's advice before choosing an output. As mentioned before, observing the action recommended by the controller advice may provide the player with additional information about the joint action sampled. Therefore, players are recommended to follow $D$, but have the option to deviate from $D$ through their own strategies. When Player~$i$ decides to follow a strategy $\pi_i$ instead of $D$, we denote this by the pair $\langle D, \pi_i \rangle$. Note that once the behavior of each player is fixed, either through the specification of $D$ or by a pair of $D$ and a strategy, it becomes possible to calculate the expected payoff that each player receives in $G$ through a Markov Chain construction. This construction is given in technical detail in Section~\ref{chainsection}. We denote this expected payoff by $\rho_i(\cdot)$, and so the expected payoff that Player~$i$ receives when all players follow $D$ is given by $\rho_i(D)$. When Player~$i$ deviates from $D$ to play $\pi_i$, the expected payoff he receives is given by $\rho_i( D, \pi_i)$. This brings us to the definition of the \emph{correlated equilibrium}.

\begin{definition}[Correlated Equilibrium]
    Given a game $G$, a controller advice $D$ is a \emph{correlated equilibrium} if it is the case that for each player $i \in \Pi$, we have $\forall \pi_i \in \Pi_i. \rho_i(D) \geq \rho_i(D,\pi_i)$. 
\end{definition}

Intuitively, this means that no player can unilaterally deviate from $D$ and increase their expected payoff. Note, however, that the expected payoffs calculated are with respect to the initial state $q_{\sf init}$ in $G$. We now broaden the definition of correlated equilibria to create the \emph{subgame-perfect correlated equilibrium}, akin to the development of the \emph{subgame-perfect equilibrium} after the \emph{Nash equilibrium} in game theory~\cite{Osborne1994}. A subgame-perfect equilibrium considers \emph{subgames} that are induced by \emph{histories}.

\begin{definition}[History]\label{historydef}
    Given a game $G$, a \emph{history} $h \in (Q \times A^{\Pi})^* \times Q$ is a finite sequence of state-joint action pairs $q_{\sf init} a_0 q_1 a_1 \ldots$ such that : {\bf (1)} The first element of a history is the initial state $q_{\sf init}$. {\bf (2)} For every consecutive state-action-state triple $q_k a_k q_{k+1}$, we have $\delta(q_k,a_k)[q_{k+1}] > 0$.{\bf (3)} For every state-action pair $q_k a_k$, $a_k$ contains no actions restricted by $\mathcal{A}(q_k)$. Since $a_k \in A^{\Pi}$, $a_k = \langle a_1 \ldots a_{|\Omega|} \rangle$ and for all $ i \in \Omega$ we have $a_i \in \mathcal{A}(q_k, i)$.
\end{definition}

Intuitively, a history represents a ``plausible'' partial execution of the game $G$. All executions start at state $q_{\sf init}$. From there, all joint actions $a_0 \in A^{\Pi}$ not restricted by $\mathcal{A}$ could potentially be chosen, and the next state $q_1$ must be a state that it is possible to transition to with positive probability upon playing $a_0$ at $q_{\sf init}$. This process then repeats a finite number of times. The game's position after observing a history is called a \emph{subgame}, and, for notational convenience, we specify that a history ends with a state, not a joint action.

\begin{definition}[Subgame]\label{subgamedef}
    Given a game $G$ and a history $h$, the subgame $G|_h$ represents the game $G$ after the partial execution $h$ has occurred. For a history with last element $q \in Q$, this means that $G$ starts in state $q$ and the payoff of players that have visited their reachability goal in $h$ is fixed to $1$.
\end{definition}
These definitions of history and subgame are standard in the game-theoretic literature~\cite{Osborne1994}, but take on a more specific interpretation in our setting. Since each player has a reachability goal, a subgame induced by a history $h$ with last element $q$ is equivalent to the game $G$ started in state $q$ that remembers information about which players have already visited their reachability goal in the history. This additional information about players who have already reached their goals is further developed in Section~\ref{subgamedeviationsection}.

\begin{definition}[Subgame-Perfect Correlated Equilibrium]
    Given a game $G$, a controller advice $D$ is a \emph{subgame-perfect equilibrium} if it is the case that for each player $i \in \Pi$ and all histories $h$, we have $\forall \pi_i \in \Pi_i. \rho_i(D) \geq \rho_i(D,\pi_i)$ in $G|_h$.
\end{definition}

Intuitively, correlated equilibria consider the expected payoff that players receive when the game $G$ is started in the original starting state $q_{\sf init}$, whereas the subgame-perfect correlated equilibria consider the expected payoff that players receive in all subgames induced by changing the starting state. This definition is motivated by the well-known strengthening of the Nash equilibrium to the subgame perfect equilibrium in the game-theory literature~\cite{Osborne1994}. Note that it is possible to change the starting state to a state that is not reachable with positive probability when $D$ is followed from $q_{\sf init}$; nevertheless, this is still a valid subgame in the traditional game-theoretic sense~\cite{Fudenberg13}. Although correlated equilibria and subgame-perfect equilibria are standard topics in game theory, they are not as mainstream as Nash equilibria and subgame-perfect equilibria. For this reason, we provide some illustrative examples to assist readers looking for more intuition in Section~\ref{examplesection} of the Appendix. We now introduce the main decision problems considered in this paper.

\begin{definition}
    The \emph{correlated equilibrium verification problem} takes as input a game $G$ and a controller advice $D$ and decides whether $D$ is a correlated equilibrium in $G$. The \emph{subgame-perfect correlated equilibrium verification problem} takes as input a game $G$ and a controller advice $D$ and decides whether $D$ is a subgame-perfect correlated equilibrium in $G$.
\end{definition}

An overview of how this paper proceeds can be found in Section~\ref{outlinesection} in the Appendix.

 
\section{Constructions}\label{constructionssections}
 
The two constructions below are given with respect to a fixed Player~$i$ and are \emph{never built explicitly}. Their entries are queried on-the-fly as needed by algorithms defined in later sections.
 
\subsection{The Markov Chain \texorpdfstring{$G \times D$}{G x D}}
\label{chainsection}
We first define a Markov chain that is used to compute the expected reward for Player~$i$ when all players follow $D$. When all players play according to $D$, the interaction of $G$ and $D$ yields the Markov Chain
$G \times D = \left\langle\, Q \cup (Q \times A_i),\; q_{\sf init},\;  E,\; R_i \right\rangle$, where $E$ is the transition probability matrix of $G \times D$.  The state space augments $Q$ with pairs $Q \times A_i$ to record which action $D$ recommended to Player~$i$. This extra structure is not strictly necessary here, but it is reused
in the MDP construction below.  Transitions are defined as follows.

\begin{itemize}
  \item \textbf{From $q \in Q$:}
    State $q$ transitions to $\langle q, a\rangle \in Q \times A_i$ with
    probability
    \[
      P_{q,a} \;:=\;
        \sum_{d \in D(q), d[i] = a } \mathbb{P}(D(q) \sim d),
    \]
    the marginal probability that Player~$i$ is recommended action~$a$
    at state~$q$, where the notation $D(q) \sim d$ is the probability that the distribution $D(q)$ samples $d$.
 
  \item \textbf{From $\langle q,a\rangle \in Q \times A_i$:}
    This state transitions to $q' \in Q$ with probability
    \[
      \frac{\displaystyle\sum_{d \in D(q), d[i] = a}
            \mathbb{P}(D(q)\sim d)
            \cdot\delta(q,d,q')}{P_{q,a}}.
    \]
\end{itemize}
 
\noindent
Every entry of $E$ is a ratio whose numerator arises from a sum of at most $T$ (a derived bound on the number of rows in $D(q)$, see Figure~\ref{controllerrepresentation}) products of two $\ell$-bit numerals (giving length $O(\ell+\log T)$) and whose denominator is a sum of at most $T$ such numerals (length $O(\ell+\log T)$), and both are polynomial in the size of $G$ and $D$.  Each bit of an entry is computed by iterating over the rows of the table $D(q)$ with a logarithmic-size counter.
 
\begin{lemma}\label{unprunedchainprobs}
  The $k$-th bit of an entry $E(q,q')$ in $G \times D$ can be computed in
  \textup{LOGSPACE}.  Furthermore, $E(q,q')$ has a polynomial-size rational
  representation with respect to $G$ and $D$.
\end{lemma}
 
\subsection{The Markov Decision Process \texorpdfstring{$G_i \times D$}{Gi x D}}
We now extend the previous construction into an MDP that models possible unilateral deviations by Player~$i$. Formally, we construct the MDP
$G_i \times D =
  \langle Q \cup (Q \times A_i), q_{\sf init},
  A_i, \tau, \mathcal{A}, R_i \rangle$.
The state space and initial state are as in $G \times D$.  States in $Q$ are \emph{uncontrolled} and transition to $\langle q,a\rangle$ with probability $P_{q,a}$ as before.  States $\langle q,a\rangle \in Q \times A_i$ are \emph{controlled}: 
when the agent chooses $a^{*} \in A_i$ (modeling a possible deviation by Player~$i$), the
probability of moving to $q' \in Q$ is
\[
  \tau\!\bigl(\langle q,a\rangle,\, a^*,\, q'\bigr)
  \;=\;
  \frac{\displaystyle\sum_{d \in D(q),d[i]=a}
        \mathbb{P}(D(q)\sim d)
        \cdot\delta\!\bigl(q,\,d[i\mapsto a^*],\,q'\bigr)}{P_{q,a}},
\]
where $d[i \mapsto a^*]$ denotes the joint-action tuple obtained from $d$
by replacing its $i$-th component with $a^*$.  Action restrictions from
$\mathcal{A}$ carry over: an action $a'$ illegal for Player~$i$ at state~$q$
in $G$ is illegal at every $\langle q,a\rangle \in Q \times A_i$.  By the
same justification as Lemma~\ref{unprunedchainprobs}, we get:
 
\begin{lemma}\label{unprunedMDPprobs}
  The $k$-th bit of an entry $\tau(\langle q,a\rangle,a^*,q')$ can be
  computed in \textup{LOGSPACE}, and this entry has a polynomial-size rational
  representation with respect to $G$ and $D$.
\end{lemma}
 
\noindent
Finally, we recall that MDPs with reachability goals always admit an {\em optimal strategy}, which can always be chosen to be deterministic and memoryless (of type $(Q \times A_i) \rightarrow A_i$), so it suffices to reason about those strategies rather than the full strategy space of Definition~\ref{strategydef}~\cite{Puterman94}.

\section{Computing Payoffs Under \texorpdfstring{$D$}{D}}\label{payoffcomputesection}

In this section, we detail how to compute the expected payoff that Player~$i$ receives when all players follow the controller advice $D$ in $G$. To make this computation applicable to both correlated equilibria and subgame-perfect correlated equilibria, we show how to compute the expected payoff when the game start from an arbitrary state. To compute the relevant expected payoff received by Player~$i$ for correlated equilibria, we start the game in state $q_{\sf init}$. For subgame-perfect correlated equilibria, we start the game at the last state of a history -- a point that is further discussed in Section~\ref{deviationsection}. 

Therefore, the main goal of this section is to answer the following problem algorithmically: Given a game $G$, a controller advice $D$, and a state $q \in Q$, what is the expected payoff that Player~$i$ receives in $G$ when every player follows $D$ and the game $G$ is started in state $q$ (which may not be $q_{\sf init}$)? In this section, we demonstrate that this question can be addressed through highly efficient linear algebraic techniques.

\subsection{Processing the Markov Chain \texorpdfstring{$G \times D$}{G x D}}\label{chainprocesssection}
Markov Chains can be analyzed through the use of linear-algebraic techniques~\cite{Puterman94}. Specifically, we wish to calculate the \emph{hitting probability} of the set $R_i$ from the starting state $q \in Q$.  To uniquely solve for the hitting probability via linear-algebraic techniques, specifically solving a system of linear equations, the chain $G \times D$ must be processed first. The broad idea is to construct a new Markov Chain $G' \times D$ that is \emph{weakly-connected} and \emph{weakly-chained substochastic}~\cite{linalg}. These two properties allow us to analyze $G \times D$ on-the-fly, giving us a subpolynomial bound. The weakly-connected property states that every state in the chain must be reachable from the start state. The weakly-chained substochastic property says that every state must be able to reach a state that has outgoing transition probabilities that sum to strictly less than $1$. We also make some changes that are mathematically less significant, mainly intended to improve the presentation of our algorithms.

Our first step in processing $G \times D$ involves the introduction of two predicates $W_i$ and $L_i$ over the state space $Q \cup (Q \times A_i)$. {\bf (1)} The predicate $W_i$ is the set of states that reach $R_i$ with probability $1$. {\bf (2)} In the same manner, $L_i$ is the set of states that reach $R_i$ with probability $0$. Furthermore, $L_i$ also includes states that are not reachable from $q$ with positive probability.

\begin{lemma}\label{WiLilemma}
    The problem of deciding whether a state belongs to $W_i$ or $L_i$ is in NLOGSPACE.
\end{lemma}

The proofs of Lemmas~\ref{WiLilemma} and~\ref{prunedchainprobqueries} can be found in the Appendix, in Section~\ref{chainprocessesingAppendix}. In order to obtain $G' \times D$ from $G$, we consolidate all states in $W_i$ into a single target state with no outgoing transitions. All states in $L_i$ are deleted. Note that since $L_i$ contains all states that are not reachable from our starting state $q$, removing the states in $L_i$ gives us our \emph{weakly-connected} condition. Once again, we do not explicitly construct the edge relation $E'$ or $G' \times D$, but we must prove the same crucial property as in Lemma~\ref{unprunedchainprobs}.

\begin{lemma}\label{prunedchainprobqueries}
    Given two states $q,q'$ in $G' \times D$ and an index $k$, the $k$-th bit of the probability $q$ transitions to $q'$ in $G \times D$ can be computed in NLOGSPACE.
\end{lemma}






Note that the probability of reaching the target state in $G' \times D$ is directly related to the expected payoff in $G \times D$, which itself accurately models the expected payoff given to Player~$i$ in $G$ when all players follow $D$. This can be seen by noting that the only changes to $G' \times D$ from $G \times D$ include consolidating the set of ``good'' states and deleting both irrelevant states and states from which the payoff given to Player~$i$ is guaranteed to be $0$.  There is one small caveat to this argument -- it may be that the only way to start from state $q$ is to pass through a state $r \in R_i$, thus ensuring that in $G \times D$, whenever $q$ is reached, the payoff to Player~$i$ will automatically be $1$. We discuss this point further in Section~\ref{subgamedeviationsection}.

As in Lemma~\ref{unprunedchainprobs} and Lemma~\ref{unprunedMDPprobs}, we would also like to reason about the size of the transition probabilities in $G' \times D$ if they were represented by ratios. The interesting case in this setting is when we are considering a transition to a state in $W_i$. There are at most $k = |Q \cup (Q \times A_i)|$ states in $W_i$, and therefore the sum can have at most $k$ summands. By Lemma~\ref{unprunedchainprobs}, we know that these summands have a numerator of length at most $n = 2\ell + \log(T) + 1$ and a denominator of length at most $d = \ell \log(T) + 1$, both of which are polynomial in $G$ and $D$. In order to perform addition on the $k$ ratios, a common denominator must be found for all of the ratios. In the worst case, this involves multiplying all of the denominators together.  When two numerals of length $a$ and $b$  are multiplied, the result is a numeral of length at most $a + b$. Therefore, the denominator grows linearly with each multiplication, resulting in a polynomial-sized common denominator overall of size $O(kd)$. For the numerators, in the worst case, each is multiplied by the common denominator before they are summed together. Therefore, the numerator has a maximum size of $O(n+kd + \log(k))$, which is once again polynomial in the size of $G$ and $D$. This gives us the same structure as in Lemma~\ref{unprunedchainprobs} -- a polynomial-size ratio that can be queried in polylog space.

\begin{lemma}\label{prunedchainprobrepresentations}
    Given two states $q,q' \in Q \cup (Q \times A_i)$, $E'(q,q')$, when represented as a ratio, has polynomial size with respect to $G$ and $D$.
\end{lemma}

Furthermore, note that $G' \times D$ must be \emph{weakly-chained substochastic} -- every state that was kept in $G' \times D$ must be able to reach the target state with positive probability, and the target state has no outgoing transitions.

\subsection{Computing Payoffs}\label{payoffcomputationdetailssection}
Our objective is to now compute the hitting probability of the target state from some initial state in $G' \times D$.  We can do this by summing over the mutually exclusive events that the target state is reached in exactly $i$ steps. Formally, we have the expression $ \sum_{i=0}^{\infty} ({\sf start} \cdot E'^{i} \cdot {\sf target})$.

If we let $S$ represent the size of the pruned state space of $G'\times D$, then ${\sf start}$ is an $S \times 1$ row vector Dirac distribution that assigns probability $1$ to the initial state and $0$ on all other states and ${\sf target}$ is the $1 \times S$ column vector Dirac distribution that assigns probability $1$ to the target state and $0$ elsewhere. Finally, $E'$ is the $S \times S$ transition matrix of $G' \times D$. Since all terms are non-negative, we can rewrite the sum as
$  {\sf start} \cdot (\sum_{i=0}^{\infty}E'^{i}) \cdot {\sf target}$.
As discussed before, $E'$ is weakly-chained substochastic. This implies that the spectral radius of $E'$, $\rho(E')<1$ (Corollary 2.6,~\cite{weaklychainedhassmallspectral}). Therefore $(\sum_{i=0}^{\infty}E'^{i})$ is a Neumann series, the inverse $(I - E)^{-1}$ exists and $(\sum_{i=0}^{\infty}E'^{i}) = (I - E')^{-1}$~\cite{linalg}. Rewriting, we have 
$  {\sf start} \cdot (I-E')^{-1} \cdot {\sf target}$

From a complexity theoretic viewpoint, the dominating operation in this expression is matrix inversion, which belongs to the class NC\textsuperscript{2}~\cite{csanky76} for explicitly represented matrices. The classes NC\textsuperscript{2} and NLOGSPACE belong to the class LOG\textsuperscript{2}SPACE~\cite{johnson-chapter-handbook-of-tcs}. Let us then imagine using a LOG\textsuperscript{2}SPACE machine $\mathcal{M}$ to compute the bits of $(I-E)^{-1}$ (since LOG\textsuperscript{2}SPACE is a class of decision problems, we consider the problem of deciding a single bit). In the most standard setting, such a machine would have a polynomial-sized read-only tape along with a LOG\textsuperscript{2}SPACE-sized work tape. Since the bits of the explicit representation of $E'$ can be computed in LOGSPACE (Lemma~\ref{prunedchainprobqueries}) and the explicit form of $E'$ where probabilities are represented as ratios is polynomial in size w.r.t $G$ and $D$ (Lemma~\ref{prunedchainprobrepresentations}), the read-only input tape can be virtualized. In order for $\mathcal{M}$ to compute the bits of $(I-E)^{-1}$, it virtualizes the head of the read-only tape by storing its position through a  LOGSPACE counter (possible due to the polynomial-sized representation of $E'$) and computing its contents on-the-fly using a separate LOGSPACE worktape. Therefore, it is possible to create a LOG\textsuperscript{2}SPACE Turing machine $\mathcal{M}$ that can decide bits of $ {\sf start} \cdot (I-E')^{-1} \cdot {\sf target}$.

\begin{lemma}\label{payoffquery}
    Given an index $k$, the $k$-th bit of the expected payoff that Player~$i$ receives in $G$ when started from a state $q$ and all players follow the controller advice $D$ can be decided in LOG\textsuperscript{2}SPACE.
\end{lemma}

Furthermore, by Lemma~\ref{prunedchainprobrepresentations}, the entries of $(I-E')$, when represented as ratios, have polynomial size w.r.t. $G$ and $D$. Then, by Corollary 3.2a of~\cite{schrijver1998}, the entries of $(I-E')^{-1}$, when represented as ratios, have polynomial size with respect to $G$ and $D$. Therefore, we can conclude that $  {\sf start} \cdot (I-E')^{-1} \cdot {\sf target}$ can be represented by a ratio of polynomial size in $G$ and $D$. The complete picture is then that we have a polynomial-size representation that can be queried in polylog space.

\begin{lemma}\label{payoffsize}
    The expected payoff of Player~$i$ in the game $G$ obtained when all players follow the controller advice $D$ can be represented as a ratio of polynomial size in $G$ and $D$.
\end{lemma}

\section{Existence of a Deviation}\label{deviationsection}

As outlined in Section~\ref{outlinesection}, computing the expected payoff a player receives when all players follow the controller advice $D$ is only half the puzzle; it must also be determined whether this value is optimal. In this section, we tackle the scenarios that arise when Player~$i$ is allowed to deviate from $D$. Our goal is to characterize whether or not a profitable deviation for Player~$i$ exists. It is important to note that this does not necessarily imply that we \emph{must} solve the MDP $G_i \times D$ in order to obtain its optimal value; indeed, whether we must solve $G_i \times D$ or not plays a critical role in distinguishing the correlated equilibrium verification problem and the subgame-perfect correlated equilibrium verification problem.

\paragraph*{Subgame-Perfect Correlated Equilibria}\label{subgamedeviationsection}

In order to analyze the relationship between subgames and states in $G \times D$, it is important to address the point that, unlike infinite-horizon goals such as the parity condition, the reachability condition depends on the observed history.  For a Player~$i$, we now introduce the terminology of \emph{relevant subgames}, which are induced by \emph{relevant histories}.

\begin{definition}[Relevant History]\label{relhistorydef}
    Given a Player~$i$ in a game $G$, a \emph{relevant} history is a history 
    $h \in (Q \times A)^* \times Q$ that never visits a state in $R_i$.
\end{definition}

Intuitively, this means that the history does not contain a pair $\langle q ,a \rangle  \in Q \times A^{\Pi}$ such that $q \in R_i$ nor an element $q \in Q$ with $q \in R_i$. The idea is that for Player~$i$, histories that already visited $R_i$ are ``irrelevant'' from a deviation standpoint since they already achieved his reachability goal and guaranteed him his maximal payoff. This naturally leads us to a definition for relevant subgames.

\begin{definition}[Relevant Subgame]
    Given a Player~$i$ in a game $G$, a relevant subgame is a subgame $G|_h$ (Definition~\ref{subgamedef}) induced by a history $h$ that is relevant to the Player~$i$.
\end{definition}

The idea is then that potential deviations by Player~$i$ in a relevant subgame induced by a relevant history with last element $q \in Q$ can be analyzed by starting the Markov Decision Process $G_i \times D$ in state $q$. This is relatively straightforward -- a relevant subgame represents a plausible partial execution of the game that did not already visit $R_i$, and so Player~$i$ is still motivated to maximize his probability of visiting $R_i$. Therefore, it is useful to refer to relevant subgames 
by the last state the corresponding history visits, i.e., the relevant subgame $G|_h$ for the relevant history $h$ with last element $q \in Q$ is also referred to as $G_q$.  Determining whether a state $q$ can be the last state of some relevant history can be done in NLOGSPACE, as the history can be guessed on-the-fly. In such cases, we call the state $q$ \emph{relevant} as well, as there is some subgame that incentivizes Player~$i$ to maximize his probability of reaching $R_i$ when the game is started in state $q$. If there is no history that reaches $q$ without also visiting $R_i$, (for example, a state $q \in R_i$), then $q$ is \emph{irrelevant} -- every time Player~$i$ is in state $q$, they have already met their goal and no longer have a preference on further executions of the game. 

The main point of this section is to establish Theorem~\ref{existencedeviationsubgame}, which equates a controller advice $D$ not being a subgame-perfect correlated equilibrium in a game $G$ to a highly local condition, which we call a ``one-step'' deviation to reflect the fact that $D$ is only deviated from at a single state.

\begin{theorem}[One-Step Deviation]\label{existencedeviationsubgame}
Given a controller advice $D$ in a game $G$, suppose that $D$ is not a subgame-perfect correlated equilibrium in $G$, i.e. there is some Player~$i$ and some relevant subgame $G_{q^*}$ such that Player~$i$ is incentivized to deviate in $G_{q^*}$. Then, in the MDP $G_i \times D$, there exists a state $\langle q, a \rangle \in (Q \times A_i)$ 
that satisfies the following:
{\bf (1)} $q$ is relevant. {\bf (2)} $\langle q ,a \rangle \in  (Q \times A_i)$ is reachable from $q \in Q$ with positive probability. {\bf (3)} There exists an action $a^* \in A_i$ such that by choosing the action $a^*$ instead of the action $a$ at the state $\langle q,a \rangle$ and following the recommendations of $D$ everywhere else (playing the action $a'$ at a state $\langle q',a' \rangle$) in the MDP $G_i \times D$, the hitting probability of the set $G_i$ from $\langle q, a \rangle$ is strictly larger than the hitting probability of $G_i$ from $\langle q,a \rangle$ in $G \times D$.
\end{theorem}
The proof of this theorem can be found in Section~\ref{onestepsection}.

\paragraph*{Correlated Equilibria}\label{correlateddeviationsection}
In this paragraph, we briefly discuss why the argument of Theorem~\ref{existencedeviationsubgame} does \emph{not} apply to the standard correlated equilibrium. In Theorem~\ref{existencedeviationsubgame}, we guaranteed the existence of \emph{some} witness, but did not specify \emph{at which state} the hitting probability was improved from. Therefore, we cannot conclude that the logic of Theorem~\ref{existencedeviationsubgame} would produce an improvement at $q_{\sf init}$ specifically.  In Section~\ref{correlatedcomplexitysection}, we demonstrate that this difference is pathological and, under standard complexity-theoretic assumptions, seperates the complexity of the subgame-perfect correlated equilibrium verification problem and the correlated equilibrium verification problem.

\section{Computational Complexity}\label{complexitysection}

\subsection{Subgame-Perfect Correlated Equilibria}\label{subgamecomplexitysection}

By Theorem~\ref{existencedeviationsubgame}, we know that if a controller advice $D$ is not a subgame-perfect correlated equilibrium in $G$, then this is witnessed by a ``one-step'' deviation from a single state $\langle q,a \rangle \in G_i \times D$. The idea is now to construct two different Markov Chains. Markov Chain $A$ is $G \times D$ started in state $q$. Markov Chain $B$ is $G \times D$ started in state $q$ with the one-step deviation at $\langle q,a \rangle$ applied. That is to say, $D$ is followed at every state of $B$ except $\langle q,a \rangle$, and at this state the action $a^*$ is chosen instead of $a$. This single state $q$ that both chains start in, and the one-step deviation change in $B$ can be guessed in NLOGSPACE. Furthermore, it can be verified that $q$ is relevant and $\langle q,a \rangle$ is reachable from $q$ with positive probability in NLOGSPACE as well.

Our goal is now to compute the hitting probability from $q$ in $A$, which we call $A_q$, and compare it to the hitting probability from $q$ in $B$, which we call $B_q$. We know that if $D$ is not a subgame-perfect correlated equilibrium, then such a $q$ must exist where the hitting probability in $B$ must be strictly larger. As shown in Lemma~\ref{payoffquery} the bits of each of the hitting probabilities can be computed in LOG\textsuperscript{2}SPACE. The idea is then to compare both probabilities lexicographically to show that the hitting probability in $B$ is higher than $A$. Therefore, the idea is to keep track of a counter $k$ and query the $k$-th bit of both $A_q$ and $B_q$.

These hitting probabilities, however, may be non-terminating numerals, meaning that the counter $k$ may grow without bound. By Lemma~\ref{payoffsize}, if we represent these hitting probabilities as ratios, then they have size polynomial in $G$ and $D$. Denote $A_q = \frac{n_1}{d_1}$ and $B_q = \frac{n_2}{d_2}$. If these two numbers are not equal, then they must differ by at least $\frac{1}{d_1 \cdot d_2}$.
Both $d_1$ and $d_2$ have polynomial size w.r.t $G$ and $D$, so the product $d_1 \cdot d_2$ has polynomial size w.r.t $G$ and $D$.  In terms of value, this means that the difference $|A_q - B_q|$  is at least $2^{-p(|G|,|D|)}$ for some polynomial $p$. Therefore, if we assume that $A_q$ and $B_q$ are different and compare their numeral expansions bit-by-bit, we must find a difference before the $p(|G|,|D|) + 1$-st bit. This means that we can compare $A_q$ and $B_q$ by checking the first $p(|G|,|D|)$ bits only, as if $A_q$ and $B_q$ agree on the first $p(|G|,|D|)$ bits, then they must be equal. Finally, we can conclude that when the counter $k$ is represented in binary, we can represent it in LOGSPACE, giving us the last piece of our LOG\textsuperscript{2}SPACE algorithm to determine if $D$ is \emph{not} a subgame-perfect correlated equilibrium in $G$. Since LOG\textsuperscript{2}SPACE is closed under complementation~\cite{papacomplexity}, we get a final upper bound on the complexity of the subgame-perfect correlated equilibrium verification problem.

\begin{theorem}\label{subgameverificationcomplexitytheorem}
    The subgame-perfect correlated equilibrium verification problem belongs to LOG\textsuperscript{2}SPACE.
\end{theorem}

\subsection{Correlated Equilibria}\label{correlatedcomplexitysection}

In this section, we demonstrate that the correlated equilibrium verification problem is P-complete. The MDP $G_i \times D$ can be explicitly solved in polynomial time~\cite{Puterman94}. Recall that a policy is a strategy that is both deterministic and memoryless, and that in all MDPs, there is always a policy that achieves the optimal value. Therefore, this optimal value can be computed and explicitly represented in polynomial time, and it can be directly compared to the directly represented expected payoff that Player~$i$ receives in the Markov Chain $G \times D$ in polynomial time. 
\begin{theorem}\label{corrverPupperboundonly}
    The correlated equilibrium verification problem belongs to P.
\end{theorem}
We prove hardness via a reduction from the P-complete Circuit Value Problem (CVP)~\cite{papacomplexity}; the details are in Section~\ref{Preductionsection} of the Appendix.

\begin{theorem}\label{corrverificationcompletenesstheorem}
    The correlated equilibrium verification problem is P-complete.
\end{theorem}

This result contradicts the intuition that correlated equilibria should admit an easier verification problem than subgame-perfect correlated equilibria, a strengthening of the former. This is due to the fact that the correlated equilibrium concept does not admit a ``one-step'' deviation as in Theorem~\ref{existencedeviationsubgame}. In Theorem~\ref{existencedeviationsubgame}, some witness state for a local deviation was proven to exist, but it was not clear which state this was. In the correlated equilibrium verification problem, deviations must increase the value at the initial state specifically. This may require more than one local change depending on the structure of $G_i \times D$, as solving an MDP is an inherently sequential problem~\cite{PapadimitrioucomplexityofMDP}. For subgame-perfect correlated equilibria, this is no longer the case, and the fact that we can witness a deviation at an arbitrary ``relevant'' state makes the problem inherently parallelizable. This distinction between the sequential and parallel approaches is aptly illustrated by contrasting the use of NC algorithms for the subgame-perfect correlated equilibrium verification problem with the P-complete result for the correlated equilibrium verification problem. 


\section{Succinct Representations via Bayesian Belief Networks}\label{bayesiannetworksection}\label{BNdefsections}

The complexities of both the correlated equilibrium verification problem (Theorem~\ref{corrverificationcompletenesstheorem}) and the subgame-perfect correlated equilibrium verification problem (Theorem~\ref{subgameverificationcomplexitytheorem}) were heavily influenced by the representation of the transition table and the controller advice. Specifically, we represented both as explicit tables. A natural next question is then to study the complexities of both problems when a more succinct representation is used. For this, we turn to the popular concept of a Bayesian Network (BN), a graphical model used to represent probability distributions in a succinct manner~\cite{principlesofmodelcheckingbook,bayesiannetworkinference}.

A Bayesian Network is a triple $\langle V, E, \Theta \rangle$~\cite{bayesiannetworkinference,bayesianinferencecomplexity}, where $V$ is a set of vertices, $E$ a set of directional edges, and $\Theta$ a set of prior and conditional probabilities. Together, the pair $\langle V,E \rangle$ forms a directed acyclic graph. The set $\Theta$ associates each vertex in $V$ to either a \emph{prior distribution}, if it has no incoming edges, or a \emph{conditional distribution}, if it does. Intuitively, this allows even distributions with a large support to be represented succinctly, as we illustrate with an example in Section~\ref{bayesianexamplesection}.

In this section, the BN model is used to represent both the controller advice $D$ and the transition table $\delta$.  Just as in Section~\ref{definitionsection}, we start by representing the controller advice, which is given in a curried format  -- one $D_q$ for each state $q \in Q$. This BN has a vertex $A_i$ associated with each player. The idea is that the probability that this variable takes a specific value such as $1$ is the same probability that the controller advice recommends Player~$i$ the action $1$. In order to accomplish this, we allow the controller advice to assign prior distributions to some of the $A_i$ vertices and then use conditional distributions and extra computational vertices to assign distributions to the rest. By allowing for these general computations, we allow for succinct representations, as demonstrated in Figure~\ref{BNfigure} in Section~\ref{bayesianexamplesection}.

We now move on to the transition function $\delta$, which we assume is provided in the same curried manner as the controller advice -- there is one transition function $\delta_q : A^{\Pi} \rightarrow {\sf Dist}(Q)$ for each state. Unlike the controller advice, the function $\delta_q$ has a set of inputs $A^{\Pi}$ even after currying. We therefore extend the Bayesian Network framework slightly by extending the definition of $\Theta$. Previously, each vertex $v \in V$ was associated with either a prior or a conditional distribution under $\Theta$. Now, we allow for an additional third possibility -- that a vertex $v \in V$ can be simply designated ``input'', and not associated with a distribution.  In order to query the probability that a given state $q_i \in Q$ is the next successor state from a state $q$ upon a given joint action input, $\delta_q$ contains $|Q|$ ``output'' vertices such that the probability that $q_i$ is the next state sampled in $G$ when the joint action $\overline{a} \in A^{\Pi}$ is chosen from state $q$ is equal to the probability that the vertex $q_i$ takes the value $1$ in the BN $\delta_q$ when the input vertices are instantiated to the values of $\overline{a}$. This general framework is illustrated in Figure~\ref{BNfiguretrans}.  

\begin{figure}[ht]
\centering
\begin{tikzpicture}[
  scale=0.5, every node/.style={scale=0.5},
  arr/.style={->, >=stealth, thick, color=blue!60!black},
  inp/.style={
    circle, draw, thick,
    minimum size=1.0cm,
    fill=blue!30!gray!50,
    text=black,
    font=\small\itshape,
    inner sep=2pt
  },
  outnode/.style={
    circle, draw, thick,
    minimum size=1.0cm,
    fill=green!40!gray!40,
    text=black,
    font=\small\itshape,
    inner sep=2pt
  },
]


\node[inp] (A1) at (2.0,  0) {$A_1$};
\node[inp] (A2) at (4.5,  0) {$A_2$};
\node[inp] (A3) at (7.0,  0) {$A_3$};
\node[inp] (Ak) at (12.0, 0) {$A_k$};

\node at (9.7, 0) {$\cdots$};

\node[draw, thick, minimum width=1.5cm, minimum height=0.7cm,
      font=\small, fill=white] at (0.2, 0) {Input};


\node[draw, very thick, minimum width=11.5cm, minimum height=1.8cm,
      font=\large, fill=white, align=center]
  (compbox) at (7.0, -2.5) {Computational Vertices};


\draw[arr] (A1.south) -- (A1.south |- compbox.north);
\draw[arr] (A2.south) -- (A2.south |- compbox.north);
\draw[arr] (A3.south) -- (A3.south |- compbox.north);
\draw[arr] (Ak.south) -- (Ak.south |- compbox.north);


\node[outnode] (q1) at (2.0,  -5.2) {$q_1$};
\node[outnode] (q2) at (4.5,  -5.2) {$q_2$};
\node[outnode] (q3) at (7.0,  -5.2) {$q_3$};
\node[outnode] (qm) at (12.0, -5.2) {$q_m$};

\node at (9.7, -5.2) {$\cdots$};


\draw[arr] (compbox.south -| q1) -- (q1.north);
\draw[arr] (compbox.south -| q2) -- (q2.north);
\draw[arr] (compbox.south -| q3) -- (q3.north);
\draw[arr] (compbox.south -| qm) -- (qm.north);

\end{tikzpicture}
\caption{The layout of a transition function Bayesian Network
$\delta_q$ for a game with $|\Pi|=k$ players and $|Q|=m$ states.
The $k$ input vertices $A_1,\ldots,A_k$ feed into a block of
computational vertices, whose $m$ output vertices $q_1,\ldots,q_m$
give the probability of transitioning to each successor state.}
\label{BNfiguretrans}
\end{figure}
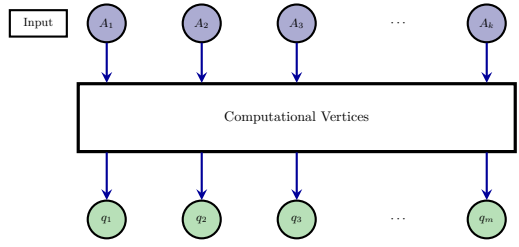

As demonstrated in Figure~\ref{BNfiguretrans},  there are $|\Pi| = k$ input vertices that are not mapped to a distribution by $\Theta$, but are rather simply labeled `input'.  Furthermore, just as we introduced new vertices to compute the value of $A_{2k+1}$ in Figure~\ref{BNfigure}, we similarly allow for intermediary computational vertices, and, by doing so, allow for succinct representations when possible.  Although the $m$ output vertices $q_1 \ldots q_m$ are associated with conditional distributions by $\Theta$, we do not specify these distributions in Figure~\ref{BNfiguretrans}, as, unlike Figure~\ref{BNfigure}, we are not working with a concrete example.

The canonical Bayesian Network model does not have ``input'' vertices. It may then seem, however, that we must iterate over all possible inputs and make $|A^{\Pi}|$ different queries on $\delta_q$ in the worst case in order to compute the probability that $q$ transitions to some other state $q'$ when the controller advice $D_q$ is followed. Fortunately, this is not the case. A key observation is that the two Bayesian Networks -- the ones representing $D_q$ and $\delta_q$ -- can be directly composed, as the input vertex $A_i$ in $\delta_q$ can easily be assigned a conditional distribution that matches its value to the vertex $A_i$ in $D$ (the same construction was used in Figure~\ref{BNfigure} to give $A_{k+i}$ the same distribution as $A_i$). This creates one single Bayesian Network without input vertices that can be queried to determine the value of $E(q,q')$, the probability needed for our Markov Chain construction in Section~\ref{constructionssections}, with one minor caveat. As opposed to a state space of $Q \cup (Q \times A_i)$, we directly construct a chain over $Q \times Q$ (recall that in Section~\ref{constructionssections}, the state space of $Q \cup (Q \times A_i)$ in the Markov Chain $G \times D$ was introduced as a notational convenience meant to minimize the notational overhead incurred when constructing the Markov Decision Process $G_i \times D$).

For the MDP $G_i \times \pi$, however,  the state space $Q \times (Q \cup A_i)$ becomes relevant once again. First, we need the probability that at a state $q$, the controller advice $D$ recommends an action $a$ to determine the probability that the MDP transitions to the state $\langle q,a \rangle$. This can be done by directly querying the controller advice $D_q$ for the probability that the vertex $A_i$, which represents player i's action, takes the value `$a$'.

We also need the transition probabilities for states in $(Q \cup A_i)$ where the agent in the MDP is allowed to choose an action, recalling that such probabilities are given by the transition function $\tau$. For this, we construct a series of transition function BNs (without input vertices) $\tau_{\langle q, a \rangle, b}$ that represent the scenarios in which the MDP is in state $\langle q, a \rangle$ but action $b$ was chosen instead of $a$. While these may seem like complex constructions, they can be constructed in a straightforward manner given $\delta_q$ and $D_q$. In order to condition on the event that action $a$ was recommended to Player~$i$, we can condition the controller advice BN $D_q$ to use the prior distribution $(\mathbb{P}(A_i = a) = 1)$. In order to model the event that Player~$i$ chose to play action `$b$' instead of action `$a$', we assign $A_i$ the value of $b$ in all subsequent computations in $\delta_q$ (a substitution that can easily be done with an intermediate variable that takes value $b$ with probability $1$). Once again, the controller advice $D_q$ can be composed with the transition function $\delta_q$ to create a single BN $\tau_{\langle q, a\rangle, b}$ without input vertices. In this way, we create a polynomial number of BNs that allow us to compute the transition probabilities given by $\tau$ in the Markov Decision Process $G_i \times D$. Furthermore, the straightforward nature of the constructions clearly implies that succinct representations are preserved -- if $D_q$ and $\delta_q$ have succinct representations, then so do all $\tau_{\langle q,a \rangle ,b}$. 

Therefore, given BN representations of $D$ and $\delta$ (which, for ease of presentation, were both curried), we can construct BN representations of both $G \times D$ and $G_i \times D$ as in Section~\ref{constructionssections}.  If we wish to unfold these representations, it should be noted that querying a BN is \#P-complete~\cite{bayesianinferencecomplexity,arorabarakcomplexity}. In order to get explicit representations of $G \times D$ and $G_i \times D$, we must query these BNs a polynomial number of times in order to establish all transition probabilities. This construction can therefore be done in P\textsuperscript{\#P}, the complexity class of polynomial time Turing machines with access to a \#P oracle, which is equal to P\textsuperscript{PP}~\cite{complexityzoo,arorabarakcomplexity}. 

Once $G \times D$ and $G_i \times D$ are constructed, the rest of the verification algorithms can proceed exactly as in the explicit table representation case, outlined in Section~\ref{outlinesection}. Both verification procedures can be decided in polynomial time (Theorems~\ref{subgameverificationcomplexitytheorem} and~\ref{corrverificationcompletenesstheorem}), giving us an overall complexity of P\textsuperscript{PP}. 

Furthermore, since querying the exact probability that a vertex takes on a value in a single Bayesian network is \#P-complete~\cite{bayesianinferencecomplexity}, the general problem of querying the bits of an output probability must be at least PP-hard~\cite{arorabarakcomplexity}. We use this for our hardness result since it is unknown whether P\textsuperscript{PP} = PP and, unlike P = NP, there does not seem to be convincing structural evidence supporting either equality or separation.

\begin{theorem}\label{BNrepcomplexitytheorem}
    When given Bayesian Network representations of $D$ and $\delta$, the subgame-perfect correlated equilibrium verification problem and the correlated verification problem are in P\textsuperscript{PP} and are PP-hard.
\end{theorem}

Note that this overall complexity is almost entirely driven by the complexity of reasoning about the Bayesian Network representations of the controller advice $D$ and the transition function $\delta$. This complexity ends up dominating the complexity of the verification problems, overshadowing the separation we witnessed before between the subgame-perfect correlated equilibrium verification problem and the correlated equilibrium verification problem. 

\section{Conclusion}\label{conclusion}

In this paper, we studied the computational complexity of the verification problem for correlated equilibria and subgame-perfect correlated equilibria in concurrent probabilistic games, under two input representations: explicit tables
(Theorems~\ref{corrverificationcompletenesstheorem}
and~\ref{subgameverificationcomplexitytheorem}) and Bayesian Networks
(Theorem~\ref{BNrepcomplexitytheorem}). Our results paint an interesting picture. Under explicit tables, the subgame-perfect variant is actually \emph{easier} than for  correlated equilibria, since its one-step deviation
property makes it inherently parallelizable. A similar observation was also made in the following
recent results presented in~\cite{RV25}. Nevertheless, this advantage disappears under the Bayesian
Network representation, where the cost of querying the network dominates, and both problems collapse to the same complexity. These findings underscore the importance of multivariate complexity analysis in equilibrium analysis: the choice of representation for the controller
advice and the transition function can be as consequential as the equilibrium concept itself.

There are two natural directions for future work. First, the $\varepsilon$-equilibrium variants: while we expect that the correlated-equilibria results will carry over largely
unchanged, we conjecture that the $\varepsilon$-subgame-perfect correlated-equilibrium verification problem is P-complete, yielding yet another counterintuitive reversal. Second, the realizability problem, which asks
whether an equilibrium \emph{exists} in a given game, remains
open~\cite{challenge}, partly because reachability objectives induce discontinuous payoff functions that block classical fixed-point arguments
(see Section~\ref{discont} in the Appendix). The verification techniques developed here, combined with non-deterministic guessing, offer a path towards realizability results in the probabilistic
setting~\cite{principlesofmodelcheckingbook}.

\bibliography{bib}

\section{Overview}\label{outlinesection}

Both correlated equilibria and subgame correlated equilibria address the question of whether a player in a game $G$ has an incentive to change their strategy and deviate from the controller advice $D$ to receive a higher expected payoff. Therefore, our algorithms proceed in two broad steps: one to analyze expected payoffs when all players follow $D$, and the other to analyze expected payoffs when one player deviates from $D$. For convenience's sake, all constructions are made with respect to a single Player~$i$. The correlated equilibrium only considers unilateral deviations made by individual players. This allows us to treat them neatly by analyzing each Player~$i$ separately.
 
In order to mathematically analyze the expected payoffs when players follow $D$ or deviate, we introduce two constructions in Section~\ref{constructionssections}. When all players commit to following $D$ in $G$, the interaction of the two creates a Markov Chain, and when one player deviates, but the rest follow $D$, the result is a Markov Decision Process (MDP). Thus, the chain enables us to analyze the case in which all players follow $D$, and the MDP enables us to analyze the case when a specific player deviates, thereby covering both scenarios. The details regarding the Markov Chain computations can be found in Section~\ref{payoffcomputesection}. In Section~\ref{deviationsection}, we characterize the conditions needed to determine whether a player has a profitable deviation from $D$ or not, which impacts whether we have to explicitly solve our MDP construction or not. This analysis culminates in an overall characterization of the complexities of the verification problems in Section~\ref{complexitysection}. Section~\ref{BNdefsections} is a stand-alone section that reconsiders the complexities of the verification problems when the inputs are represented in a succinct manner by utilizing Bayesian Networks.

\section{Examples}\label{examplesection}

In this section, we provide a few examples to help the reader distinguish between correlated equilibria, subgame-perfect correlated equilibria, and Nash equilibria. The purpose of this section is to provide the reader with simple examples meant to enhance their intuition, and, as such, the games considered in this section do not invoke the full expressive power of the probabilistic concurrent game graph model presented in Definition~\ref{pCGS}. We assume some familiarity with fundamental game-theoretic concepts such as the Nash equilibrium, as in~\cite{Osborne1994}.

The classic example used to differentiate Nash equilibria from correlated equilibria is the so-called ``Chicken Game''. While there are many instances of the Chicken Game in the game theory literature, this specific example is taken from the Wikipedia article titled  ``Correlated Equilibrium''~\cite{wikicorrelated}, the details of which the authors have independently verified.  The Chicken Game is a two-player bi-matrix game. This means that each of the two players chooses either {\bf (C)}hicken out or {\bf (D)}are simultaneously, and, depending on their joint choice, they receive the following payoffs.

\begin{table}[ht]
\centering
\begin{tabular}{c|cc}
 & $\mathrm{C}$ & $\mathrm{D}$ \\
\hline
$\mathrm{C}$ & $(6,6)$ & $(2,7)$ \\
$\mathrm{D}$ & $(7,2)$ & $(0,0)$
\end{tabular}
\caption{Payoffs in the Chicken Game}
\end{table}

In the Chicken Game, the two players are ``horizontal'' and ``vertical''. For example, if horizontal chooses (C)hicken out and vertical chooses (D)are, then the corresponding entry (2,7) records their respective payoffs: 2 to the horizontal player and 7 to the vertical player.

Intuitively, the Chicken Game is a model of the famous scene from the movie ``Rebel Without a Cause''. In this scene, two drivers are racing to the edge of a cliff. The drivers can choose to either chicken out, meaning they safely veer off to the side, or dare, meaning they commit to racing towards the cliff. If one driver dares while the other chickens out, the former has proven his bravery and receives the highest payoff of 7, while the latter receives a payoff of 2. If both chicken out, then at least they are both equally safe and cowardly, so they receive a payoff of 6. Finally, if both drivers dare, then they both race off the cliff edge -- the worst outcome, represented by the payoff of 0.

There are two pure-strategy Nash equilibria and one mixed Nash equilibrium in the Chicken Game.  The pure-strategy equilibria are $(C,D)$ and $(D,C)$. In these equilibria, the ``total expected welfare'' of both players is $2+7=9$. The mixed equilibrium is when each player chooses $C$ with probability $\frac{2}{3}$ and $D$ with probability $\frac{1}{3}$. In this equilibrium, the ``total expected welfare'' of both players is $2 \cdot (\frac{4}{9} (6) + \frac{2}{9}(7) + \frac{2}{9}(2) + \frac{1}{9}(0)) =\frac{28}{3}$. It is worth noting that each Nash equilibrium can be transformed into an equivalent correlated equilibrium by constructing a controller that matches the joint distribution created by the product of the individual strategies used in the Nash equilibrium. 

The possibility of a correlated equilibrium, however, raises another possibility. Consider a controller advice that recommends each of the tuples $(C,C), (C,D),$ and $(D,C)$ with probability $\frac{1}{3}$. When a player is recommended $C$, they recieve an expected payoff of $\frac{1}{2}(6) + \frac{1}{2}(2) = 4$, which is greater than the payoff $\frac{1}{2}(7) + \frac{1}{2}(0) = 3.5$ that they recieve when they deviate to $D$ under the assumption that other player will follow the controller advice. Likewise, when a player is recommended $D$, they receive a payoff of $7$, which is larger than the payoff of $6$ that they would receive if they were to deviate to $C$. Therefore, this controller advice is a correlated equilibrium that offers a total welfare of $2 \cdot (\frac{1}{3}(6) + \frac{1}{3}(7) + \frac{1}{3}(2)) = 10$. Note that this behavior cannot be expressed as a product of individual strategies, because the controller correlates the actions of both players to ensure that $(D,D)$ is never chosen despite $(D,C)$ and $(C,D)$ appearing with positive probability. Therefore, there is no way to express this behavior as a Nash equilibrium. Furthermore, this correlated equilibrium offers both players the highest total welfare, and it does so symmetrically. It is therefore reasonable to consider this correlated equilibrium as the optimal outcome for players in the Chicken Game, highlighting the difference between the Nash and correlated equilibrium concepts.

The difference between correlated equilibria and subgame-perfect correlated equilibria is best illustrated by the difference between Nash equilibria and subgame-perfect equilibria. This is because the differences between the two pairs of concepts are fundamentally the same. The Nash and correlated equilibrium only consider payoffs at the initial state, whereas the subgame-perfect and subgame-perfect correlated equilibrium concepts consider payoffs at every state. even if that state is not reachable from the initial state. In order to demonstrate the difference, we construct a simple two-player turn-based game called the Market Entry Game that is often used to demonstrate the concept of an ``empty threat''~\cite{Osborne1994}.

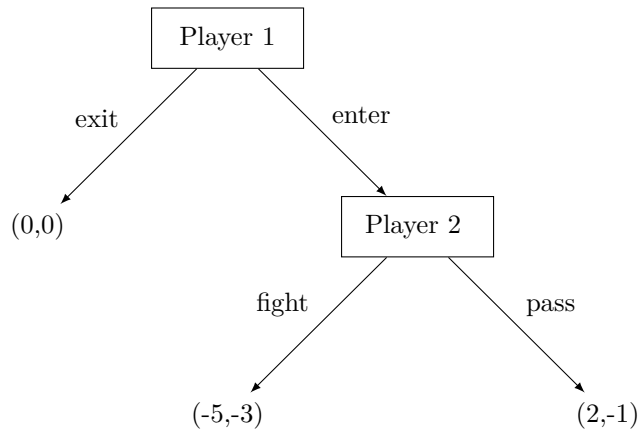
\begin{figure}[ht]
\centering
\begin{tikzpicture}[
  level distance=2.5cm,
  sibling distance=5cm,
  player/.style={rectangle, draw, minimum width=2cm, minimum height=0.8cm, align=center},
  edge from parent/.style={draw, -latex}
]

\node[player] {Player 1}
    child { 
        node { (0,0) } 
        edge from parent node[midway, above left] {exit} 
    }
    child { 
        node[player] { Player 2 }
            child { node { (-5,-3) } edge from parent node[midway, above left] {fight} }
            child { node { (2,-1) } edge from parent node[midway, above right] {pass} }
        edge from parent node[midway, above right] {enter} 
    };

\end{tikzpicture}
\caption{Payoffs in the Market Entry Game}
\label{fig:marketentry}
\end{figure}

In the Market Entry Game, player 1, a small emerging firm, first chooses whether or not to enter or exit a market. If they choose to exit, both players receive a payoff of $0$, which corresponds to the ``status quo'' being upheld. If, however, player 1 decides to enter the market, then player 2, a large, established firm, can choose whether to fight this entry or let it pass. If they choose to fight, they incur a payoff of $-3$, but the small firm suffers more, with a payoff of $-5$. If they pass, on the other hand, player 1's small firm enters the market and receives a payoff of $2$, which comes at the expense of player 2's payoff of $-1$.

The state labeled Player 1 is the initial state. The pair of actions (Exit, Fight) is both a Nash equilibrium and a correlated equilibrium (where the controller advises Player 1 to choose exit and Player 2 to choose fight). If Player 2 makes a choice, they choose the $(-5,-3)$ payoff option, so it is rational for Player 1 to choose $(0,0)$ over this penalty. Likewise, Player 2 does not get a chance to make a move, so they are also not interested in deviating. This is the key difference between the subgame-perfect equilibria concepts and their counterparts -- the tree rooted at Player 2 is now a subgame that is reached with probability 0 under the (Exit, Fight) behavior. The question now becomes whether Player 2 is required to act rationally in this subgame.

In the Nash equilibrium and the correlated equilibrium, Player 2 is not required to behave rationally (choose pass for a payoff of $-1$ over fight for a payoff of $-3$), because the game does not reach this point in (Exit, Fight). Therefore, they can ''threaten'' Player 1 into choosing exit by choosing fight. In the subgame-perfect equilibrium and the subgame-perfect correlated equilibrium concepts, they are required to behave rationally, meaning that the only rational choice is to pass. Knowing this, Player 1 chooses to enter. Therefore, (Enter, Pass) and (Exit, Fight) are both Nash equilibria and correlated equilibria, but only (Enter,Pass) is a subgame-perfect equilibrium and a subgame-perfect correlated equilibrium.

\section{The Discontinuity of Reachability Games}\label{discont}

A standard strategy to prove the existence of Nash equilibria in a multi-player game is to appeal to a fixed-point theorem, such as Brouwer or Kakutani~\cite{Osborne1994}. In order to apply such a fixed-point theorem directly to our game setting, it is necessary to have a mapping from the set of strategies to the set of payoffs that is at least upper hemi-continuous (a slightly more general condition that continuity). 

In order to see where the issue lies, consider the reachability MDP in Figure~\ref{discontMDP}. The agent (once again using our notation that the actor in an MDP is an \emph{agent}, and the actors in a game are \emph{players}) starts in state $q_{\sf init}$ and wishes to reach $q_1$. All transitions are deterministic.

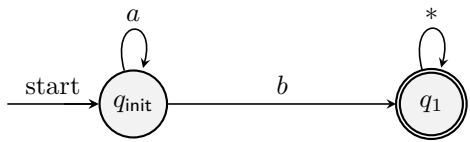
\begin{figure}[ht]
\centering
\begin{tikzpicture}[->, >=stealth, auto, node distance=3cm, semithick]

  \node[state, initial] (q0) {$q_{\sf init}$};
  \node[state, accepting, right=of q0] (q1) {$q_1$};

  \path
    (q0) edge[loop above] node{$a$} ()
         edge node{$b$} (q1)
    (q1) edge[loop above] node{$*$} ();

  \node[left=1.2cm of q0] (start) {};
  \draw[->] (start) -- node[above]{start} (q0);

\end{tikzpicture}
\caption{A simple MDP that demonstrates the discontinuities endemic to reachability games.}
\label{discontMDP}
\end{figure}
For our example, we only consider stationary strategies, i.e. ones that do not employ memory. Therefore, the probability $p$ with which the agent plays `b' at $q_{\sf init}$ determines their entire strategy. Note now that if `b' is played with probability $p > 0$, then the payoff to the agent is 1 -- the goal of $q_1$ will eventually be reached because `b' will eventually be played by the agent, and the reachability goal is evaluated over an undiscounted infinite horizon. 

Once, however, he agent plays `b' with probability $0$, then their payoff is $0$. Therefore, we can construct an infinite sequence of strategies that converge in the limit to playing `b' with probability 0, but always play `b' with probability $p>0$ and therefore receive a payoff of $1$ (for example, the $n$-th strategy for $n \in \mathbb{N}$ plays `b' with probability $\frac{1}{n}$). This discontinuity in the limit is exactly the issue that disqualifies the use of the Brouwer or Kakutani fixed-point theorems. 

Note that if we introduce a discount factor, a finite-time horizon, or both, then this discontinuity issue is resolved. Thus, these represent popular ways to consider systems with more well-understood and well-behaved mathematical properties. Our setting, however, is the undiscounted infinite-horizon setting, and so it admits these mathematical complications.

\section{Proofs of Lemmas in Section~\ref{chainprocesssection}}\label{chainprocessesingAppendix}

This section complements Section~\ref{chainprocesssection}. Here, we provide proofs for Lemma~\ref{WiLilemma} and Lemma~\ref{prunedchainprobqueries}.

\begin{manuallemma}{11}
    The problem of deciding whether a state belongs to $W_i$ or $L_i$ is in NLOGSPACE.
\end{manuallemma}
\begin{proof}
    The basic idea behind this statement is that determining whether a state belongs to $W_i$ or $L_i$ fundamentally reduces to a series of reachability queries that can be conducted in NLOGSPACE~\cite{papacomplexity}.

    For $L_i$, determining whether a state $q'$ reaches $R_i$ with probability $0$ can be done by complementing a reachability query. First, we can guess a state $r \in R_i$. Then, we can conduct a reachability query from $q'$ to $r$ to determine if there is a path that occurs with positive probability from $q$ to $r$. This procedure can be done in NLOGSPACE and determines the set of states that reach the set $R_i$ with positive probability. By complementing this algorithm in NLOGSPACE, which is possible since NLOGSPACE = co-NLOGSPACE~\cite{papacomplexity}, we get an NLOGSPACE algorithm to determine whether a state $q'$ reaches $R_i$ with probability $0$. States $q'$ that are not reachable from $q$ are even easier to determine, as this is directly a reachability query from $q$ to $q'$.

    For $W_i$, accounting for the set of states that belong to $R_i$ is trivial. The important part is then to determine if a state $q'$ that does not belong to $R_i$ reaches the set $R_i$ with probability $1$. This can be done in a similar way to $L_i$. First, we guess a state $l$ in $L_i$ that reaches $R_i$ with probability $0$ -- note that such a state must exist if it is not the case that all states reach $R_i$ with probability $1$. As described above, this property of reaching $R_i$ with probability $0$ can be determined in NLOGSPACE. Finally, we can guess a path from $q'$ to $l$ in NLOGSPACE, characterizing the set of states that do not reach $R_i$ with probability $1$. By once again using the fact that NLOGSPACE = co-NLOGSPACE, we conclude that determining whether a state that is not in $R_i$ reaches $R_i$ with probability $1$ can be done in NLOGSPACE.

\end{proof}

\begin{manuallemma}{12}
    Given two states $q,q'$ in $G' \times D$ and an index $k$, the $k$-th bit of the probability $q$ transitions to $q'$ in $G \times D$ can be computed in NLOGSPACE.
\end{manuallemma}

\begin{proof}
    Note that $E'(q,q')$ is only a meaningful query if $q \not \in W_i$ and $q \not \in L_i$. If $q \in W_i$, then it is represented by the single target state in $G' \times D$ and has no outgoing transitions to consider. If $q \in L_i$, then it was deleted and does not belong to the state space of $G' \times D$. This leaves us with three relevant cases.

    \begin{enumerate}
        \item $q' \in W_i$. In this case, all outgoing transitions from $q$ to states in $W_i$ must be summed together. Since it is possible to calculate the original bits of $E$ in LOGSPACE, the extra sum can also be accomplished in LOGSPACE by iterating over a polynomial-sized set using a logarithmically-sized counter.

        \item $q' \in L_i$. In this case, the state $q'$ does not exist in $G' \times D$ and $E'(q,q') = 0$.

        \item $q' \not \in L_i$ and $q' \not \in W_i$. In this case, the transition probability is the same as in $G \times D$, and the result from Lemma~\ref{unprunedchainprobs} gives us the needed result. 
    \end{enumerate}

    While all cases can be decided in LOGSPACE, determining which case is applicable (whether a state belongs to $L_i$, $W_i$, or neither) requires NLOGSPACE. This gives us our result. 
\end{proof}

\section{Proof of the One-Step Deviation Theorem~\ref{existencedeviationsubgame}}\label{onestepsection}

This section complements Section~\ref{subgamedeviationsection}. In it, we provide the proof of Theorem~\ref{existencedeviationsubgame}.

\begin{manualtheorem}{18}[One-Step Deviation]
Given a controller advice $D$ and a game $D$, suppose that $D$ is not a subgame-perfect correlated equilibrium in $G$, i.e. there is some Player~$i$ and some relevant subgame $G_{q^*}$ such that Player~$i$ is incentivized to deviate in $G_{q^*}$.

Then, in the MDP $G_i \times D$ there exists a state $\langle q, a \rangle \in Q \cup (Q \times A_i)$ that satisfies the following:
\begin{enumerate}
    \item $q$ is relevant.
    \item $\langle q ,a \rangle \in Q \cup (Q \times A_i)$ is reachable from $q \in Q$ with positive probability.
    \item There exists an action $a^* \in A_i$ such that by choosing the action $a^*$ instead of the action $a$ at the state $\langle q,a \rangle$ and following the recommendations of $D$ everywhere else (playing the action $a'$ at a state $\langle q',a' \rangle$) in the MDP $G_i \times D$, the hitting probability of the set $G_i$ from $\langle q, a \rangle$ is strictly larger than the hitting probability of $G_i$ from $\langle q,a \rangle$ in $G \times D$.
\end{enumerate}

\end{manualtheorem}

\begin{proof}
    Assume that the controller advice $D$ is not a subgame-perfect correlated equilibrium in $G$. By definition, there must be some Player~$i$ that is incentivized to deviate from $D$ in some subgame $G|_h$. Note that this subgame must be relevant to Player~$i$, as otherwise Player~$i$ has already guaranteed a payoff of $1$ and is not incentivized to deviate in $G|_h$.

    We refer to $G|_h$ as $G_{q^*}$. Now, consider the MDP $G_i \times D$ with initial state $q^*$. If we consider the Player~$i$ strategy $\pi^D_i$ that follows the controller advice $D$'s recommendations (at a state $\langle q, a \rangle \in Q \times A_i$ the action $a$ is always chosen), it must the case that the strategy $\pi^D_i$ is suboptimal, as otherwise Player~$i$ would not be incentivized to deviate in $G_{q^*}$. Note that the strategy $\pi^D_i$ deterministically outputs the recommendation of $D$ at the controlled states $ Q\times A_i$, and, therefore, it is both deterministic and memoryless. Definitionally, this means that $\pi^D_i$ is a policy. Furthermore, since Player~$i$ has a profitable deviation from $D$ in $G_{q^*}$, it must be the case that the hitting probability induced by $\pi^D_i$ in $G_i \times D$ with initial state $q^*$ must be suboptimal.

    Because $\pi^D_i$ is a policy that achieves a suboptimal value in an MDP with a reachability goal, we can apply one step of the \emph{policy iteration} algorithm to it. This algorithm is well-known, but readers looking for more explicit detail are referred to~\cite{Puterman94}. Namely, the policy iteration algorithm implies that there must exist a controlled state $\langle q,a \rangle \in G_i \times D$ such that changing the action from $a$, the one recommended by $D$, to some other action $a^*$ strictly increases the hitting probability at $\langle q,a \rangle$. While at least one such witness is guaranteed to exist, there are two potential caveats to address.

    The first is if a state $\langle q,a \rangle$ was not reachable from $q$ with positive probability. In a sense, this means that Player~$i$ can profitably deviate from an action that $D$ does not recommend. Note that all potential witnesses cannot be of this form, as, if they were, then $D$ would necessarily never recommend a suboptimal action and therefore, the hitting probability induced by $\pi^D_i$ at $q^*$ would necessarily be optimal, contradicting our assumption.

    The second is if we have a witness $\langle q,a \rangle$ for an irrelevant $q$. Note again that it must be the case that some $q$ must be relevant, as if $\pi^D_i$ only suggests suboptimal actions at irrelevant states, then once again $\pi^D_i$ must induce an optimal hitting probability at $q^*$. 

    Intuitively, witnesses of the first or the second category do not influence the hitting probability at $q^*$, which we know must be suboptimal. Therefore, there must be a witness $\langle q,a \rangle$ in neither category, and this witness satisfies the three conditions of our theorem. 

    Note that changing the action at the witness $\langle q,a \rangle$ does not guarantee an improvement in the hitting probability from $q^*$. Our reference to the hitting probability at $q^*$ was to make the precise statement that if all witnesses correspond to actions not recommended by $D$ or irrelevant states, then the hitting probability at $q^*$ must be optimal. It may be the case that multiple changes are needed to realize an improvement at $q^*$ specifically.

    Finally, note that after the one-step deviation is applied,  since the hitting probability at $\langle q,a \rangle$ is improved and $\langle q,a \rangle$ is reachable from $q$ with positive probability, then the hitting probability at $q$ is necessarily improved.

\end{proof}

\section{An Example of a Controller Advice Represented by a Bayesian Network}\label{bayesianexamplesection}
Consider a controller advice for some state $q$ in a game with $2k+1$ players, each of which has two actions ($a$ or $b$), that recommends tuples of the following type: The first $k$ players are assigned actions uniformly, with a probability of $\frac{1}{2^k}$ given to each tuple in $A^k$.  The next $k$ players are assigned an exact copy of the first $k$ players' assignment. The final player is assigned the bitwise XOR of the first $k$ players' choices, treating $a$ as $0$ and $b$ as $1$.

There are $2^k$ elements in the support of this distribution, representing the $2^k$ assignments for the first $k$ players. Once this assignment is fixed, it determines the assignments to the other $k+1$ players. If, however, we represented this controller advice as a table, it would have $2^k$ rows, which is exponential in the size of the input. This distribution, however, admits a compact representation as a Bayesian Network, which we illustrate in Figure~\ref{BNfigure}, which also demonstrates how $\Theta$ functions.

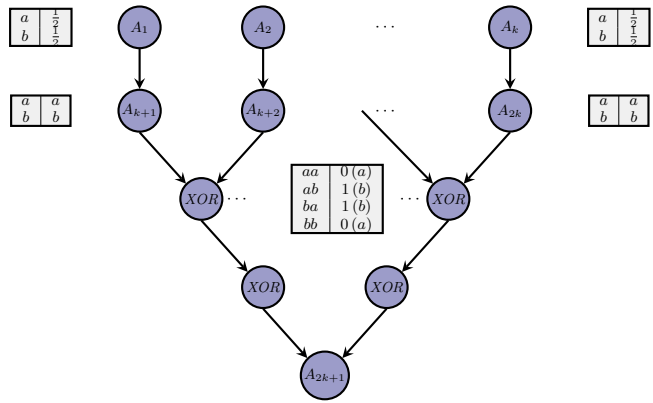
\begin{figure}[ht]
\centering
\begin{tikzpicture}[scale=0.65, every node/.style={scale=0.65},
  node distance=0cm,
  bn/.style={
    circle, draw, thick,
    minimum size=0.85cm,
    fill=blue!30!gray!60,
    text=black,
    font=\footnotesize\itshape,
    inner sep=1pt
  },
  tbl/.style={
    draw, thick, fill=gray!12,
    inner sep=0pt,
    font=\scriptsize
  },
  arr/.style={->, >=stealth, thick},
]


\node[bn] (A1) at (2.5,  0) {$A_1$};
\node[bn] (A2) at (5.0,  0) {$A_2$};
\node[bn] (Ak) at (10.0, 0) {$A_k$};

\node at (7.5, 0) {$\cdots$};

\node[tbl] at (0.5, 0) {%
  \begin{tabular}{c|c}
    $a$ & $\frac{1}{2}$ \\[1pt]
    $b$ & $\frac{1}{2}$ \\
  \end{tabular}
};

\node[tbl] at (12.2, 0) {%
  \begin{tabular}{c|c}
    $a$ & $\frac{1}{2}$ \\[1pt]
    $b$ & $\frac{1}{2}$ \\
  \end{tabular}
};


\node[bn] (Ak1) at (2.5,  -1.7) {$A_{k+1}$};
\node[bn] (Ak2) at (5.0,  -1.7) {$A_{k+2}$};
\node[bn] (A2k) at (10.0, -1.7) {$A_{2k}$};

\node at (7.5, -1.7) {$\cdots$};

\node[tbl] at (0.5, -1.7) {%
  \begin{tabular}{c|c}
    $a$ & $a$ \\[1pt]
    $b$ & $b$ \\
  \end{tabular}
};

\node[tbl] at (12.2, -1.7) {%
  \begin{tabular}{c|c}
    $a$ & $a$ \\[1pt]
    $b$ & $b$ \\
  \end{tabular}
};

\draw[arr] (A1.south)  -- (Ak1.north);
\draw[arr] (A2.south)  -- (Ak2.north);
\draw[arr] (Ak.south)  -- (A2k.north);


\node[bn] (XOR1) at (3.75, -3.5) {\textit{XOR}};
\node[bn] (XOR2) at (8.75, -3.5) {\textit{XOR}};
\node at (4.5, -3.5) {$\cdots$};
\node at (8, -3.5) {$\cdots$};

\draw[arr] (Ak1.south) -- (XOR1.north west);
\draw[arr] (Ak2.south) -- (XOR1.north east);

\draw[arr] (A2k.south)  -- (XOR2.north east);
\draw[arr] (7.0, -1.7)  -- (XOR2.north west); 

\node[tbl] at (6.5, -3.5) {%
  \begin{tabular}{c|c}
    $aa$ & $0\,(a)$ \\[1pt]
    $ab$ & $1\,(b)$ \\[1pt]
    $ba$ & $1\,(b)$ \\[1pt]
    $bb$ & $0\,(a)$ \\
  \end{tabular}
};

\node[bn] (XOR3) at (5.0, -5.3) {\textit{XOR}};
\node[bn] (XOR4) at (7.5, -5.3) {\textit{XOR}};

\draw[arr] (XOR1.south)      -- (XOR3.north west);
\draw[arr] (XOR2.south)      -- (XOR4.north east);


\node[bn] (A2k1) at (6.25, -7.1) {$A_{2k+1}$};

\draw[arr] (XOR3.south) -- (A2k1.north west);
\draw[arr] (XOR4.south) -- (A2k1.north east);

\end{tikzpicture}
\caption{The example controller advice represented as a Bayesian Network.
The first $k$ player actions $A_1, \ldots, A_k$ are chosen uniformly and
independently (prior tables on the left and right).
These are copied to players $A_{k+1}, \ldots, A_{2k}$ via identity
conditional tables (second row).
The final player $A_{2k+1}$ is assigned the bitwise XOR of the first $k$
players' choices, computed in a standard tree-like fashion via intermediate
XOR vertices (rows 3 and 4). Due to space constraints, the XOR nodes in row 4 only have one incoming arrow; in reality, they would have two. Here $a = 0$ and $b = 1$ in the XOR computation.}
\label{BNfigure}
\end{figure}

In the figure, we can see that the first $k$ player actions are chosen uniformly and independently, given by a table that details the \emph{prior} probabilities of $\frac{1}{2}$ for action $a$ and $\frac{1}{2}$ for action $b$. These actions are then carried over to Players $k+1$ to $2k$, demonstrated by the \emph{conditional} probability tables in the second row. The conditional tables state that if the action for Player~$i$ (for $i \leq k$) is sampled to be $a$, then this corresponds to an assignment of $a$ for Player $k+i$. It could also be the case that these tables also assign probabilistic distributions; for example, an assignment of $a$ for Player~$i$ corresponds to a distribution of $\frac{1}{2} a$ and $ \frac{1}{2} b$ for player $k + i$.   

The final Player, $2k+1$, was introduced to show how intermediate variables can be used in a BN to represent general calculations. Here, we compute the XOR of the first $k$ players by computing the pairwise XOR of the Players $k+1$ to $2k$ (who were necessarily assigned the same actions as the first $k$ players), and we compute this XOR in the standard tree-like fashion.  To make reading easier, we reiterate that $a=0$ and $b=1$ in this XOR computation in the table itself.  The key point is that while the table representation of this distribution is exponential in the number of players, the BN representation is polynomial in the number of players. There is a linear number of vertices with respect to the number of players (including the extra vertices required to compute the XOR), each of which is associated with a table that is at most 4 rows by 2 columns -- a fixed upper bound size.

\section{P-hardness of the Correlated Equilibrium Verification Problem}\label{Preductionsection}

This section complements Section~\ref{correlatedcomplexitysection} in the main text. In this section, we present the omitted lower bound proof showing that the correlated equilibrium verification problem is P-hard via a reduction from the Circuit Value Problem (CVP). 

\begin{problem*}[Circuit Value Problem]
The circuit value problem takes as input a \emph{circuit} with inputs, which is represented as a finite sequence of triples $C = \{ \langle a_i,b_i,c_i \rangle |  i = 1 \ldots k \}$ where $a$ is an element of the set $\{{\sf true}, {\sf false}, {\sf and}, {\sf or}\}$ and $b$ and $c$ are \emph{inputs}. Intuitively, each triple represents either a gate or an input. When $a \in \{{\sf true}, {\sf false}\}$, the triple represents an input of either true or false, and the values of $b$ and $c$ are irrelevant. When $a \in \{{\sf and}, {\sf or}\}$ then the triple represents either a logical gate with two inputs, and $b$ and $c$ represent the indices of the triples that correspond to these inputs. The circuit value problem is then to decide whether the circuit outputs $1$ for the given inputs.
    
\end{problem*}

In~\cite{PapadimitrioucomplexityofMDP}, this problem is used for P-hard reductions for a suite of problems regarding computing optimal values in MDPs (e.g., their Theorem 1). In this paper, we consider the verification problem, which involves a different setting. Therefore, we now explain how to modify the proof of Theorem 1 in~\cite{PapadimitrioucomplexityofMDP} to fit the verification setting. Before doing so, it is useful to briefly recall the proof ideas.

In order to reduce the CVP problem to computing an optimal value in an MDP, the reduction in~\cite{PapadimitrioucomplexityofMDP} followed four broad ideas.

\begin{enumerate}
    \item Two uncontrolled states labeled ${\sf true}$ and ${\sf false}$ are constructed in the MDP, which represent true and false \emph{values}. Each gate and input in the circuit was then represented by a state in the MDP. The reachability goal in the MDP is ${\sf true}$.
    \item  ${\sf or }$ gates in the circuit were mapped to controlled states in the MDP where the agent could deterministically choose between the two inputs, with the idea that at least one of the two inputs should eventually demonstrate a path to ${\sf true}$.
    \item ${\sf and}$ gates were mapped to uncontrolled states in the MDP that probabilistically transitioned to each of their input states with probability $\frac{1}{2}$, with the idea that in order for ${\sf true}$ to be reached with probability $1$ paths that reach ${\sf true}$ must be demonstrated from both input states.
    \item Each input variable deterministically maps to the state true if it was assigned the value true and false otherwise. The idea is that once the and/or gates were decomposed, they formed a distribution over the input variable states in the MDP. If a variable was assigned true, then the state representing the variable deterministically transitions to ${\sf true}$, which is the reachability goal. Otherwise, the state representing the variable deterministically transitions to ${\sf false}$. 
\end{enumerate}

It was then proven in~\cite{PapadimitrioucomplexityofMDP} that the input circuit had a value of $1$ iff the optimal value of the constructed MDP was $1$, which corresponded to forcing a final distribution that only contained variables that were assigned to true in the input. The start state of the MDP corresponded to the gate of the primary connective of the circuit.

To reduce this problem to the correlated equilibrium verification problem, we extend the MDP constructed above to a one-player game $G$ that contains one extra uncontrolled state ${\sf fp}$ (for `fixed payoff'). At the start state of the game, the player can either choose to go to ${\sf fp}$ or pass it up, at which point the MDP and the game behave identically. The state ${\sf fp}$ transitions to ${\sf true}$ with probability $ 1 - \frac{1}{2^i} + \frac{1}{2^{i+1}}$ (where $i$ is the number of gates in $C$) and ${\sf false}$ otherwise. The controller advice $D$ recommends the player to go to ${\sf fp}$, and its recommendations at other states are irrelevant and can be filled in arbitrarily.

\begin{lemma}
    The input circuit has a value of $1$ iff $D$ is not a correlated equilibrium in $G$.
\end{lemma}

\begin{proof}
    First, let us consider what happens when the input circuit $C$ does not evaluate to $1$. In the MDP, the maximal probability with which ${\sf true}$ can be reached is then at most $1 - \frac{1}{2^i}$, where $i$ is the number of gates in $C$. This can be seen by assuming that every gate is an ${\sf and}$ gate and that only one branch leads to non-acceptance. This branch occurs with probability $\frac{1}{2^d}$, where $d$ is the depth of $C$, which is underapproximated by $\frac{1}{2^i}$. Therefore, if $C$ does not evaluate to $1$, the optimal value in the MDP must be $\leq 1 - \frac{1}{2^i}$.

    In the game $G$, $D$ recommends the player to visit ${\sf fp}$. It is clear that the expected payoff of this action is $ 1 - \frac{1}{2^i} + \frac{1}{2^{i+1}}$, which is strictly greater than $ 1 - \frac{1}{2^i}$. If this value is not optimal, then the optimal action is to ignore ${\sf fp}$ and simulate the circuit. The only way for the circuit simulation to yield a better value than $ 1 - \frac{1}{2^i} + \frac{1}{2^{i+1}}$ is if the optimal value in the MDP is $1$, which is equivalent to the circuit $C$ evaluating to $1$. Therefore, in $G$, the player is incentivized to deviate from $D$ iff the circuit evaluates to $1$.

    We conclude by noting that the size of the number $ 1 - \frac{1}{2^i} + \frac{1}{2^{i+1}}$ is linear in $i$ when expressed as a binary numeral, meaning that this construction can be done in polynomial time. 
    
\end{proof}

This reduction then gives us our hardness result.

\begin{manualtheorem}{21}
    The correlated equilibrium verification problem is P-complete.
\end{manualtheorem}

\end{document}